\documentclass[11pt]{article}

\usepackage[margin=1.05in]{geometry}
\usepackage{amsmath,amssymb,amsthm}
\usepackage{graphicx}
\usepackage{booktabs}
\usepackage{algorithm}
\usepackage{algpseudocode}
\usepackage[round]{natbib}
\usepackage[table]{xcolor}
\definecolor{linkink}{RGB}{28,64,132}
\definecolor{cmpbetter}{RGB}{42,120,214}
\definecolor{cmpworse}{RGB}{227,73,72}
\usepackage[colorlinks=true,linkcolor=linkink,citecolor=linkink,urlcolor=linkink]{hyperref}
\usepackage{microtype}
\usepackage[section]{placeins}
\graphicspath{{figs/}}

\newtheorem{lemma}{Lemma}
\newtheorem{proposition}{Proposition}

\newcommand{\CD}{\mathrm{CD}}
\newcommand{\WD}{\mathrm{WD}}
\newcommand{\MD}{\mathrm{MD}}
\newcommand{\ASD}{\mathrm{ASD}}

\title{LAT: a Latinized aperiodic tiling for any sample size}

\author{Pamphile T.\ Roy\\
\small Consulting Manao GmbH, Vienna, Austria\\
\small \texttt{roy.pamphile@gmail.com}}

\date{}

\begin{document}
\maketitle

\begin{abstract}
Quasi-Monte Carlo methods allow computer experiments to be run with far
fewer simulations than crude Monte Carlo. A digital net in base two, such
as Sobol', is however only balanced when the number of samples is a
power of two, and Latin Hypercube Sampling (LHS) accepts any sample size
but only controls the one-dimensional margins. This work proposes a
space-filling design defined for any sample size, referred to as LAT for
\emph{Latinized aperiodic tiling}. The unit hypercube is cut recursively
across its longest edge following the golden section into $N$ cells of
equal volume. One point is then placed in each cell, and the margins are
made Latin while every point stays inside its own cell. The construction
only uses integer splits and costs $O(dN \log N)$. The recursion is
shown to follow the Fibonacci word and its aperiodicity is analysed. LAT
is assessed with four $L_2$-discrepancies and with the integration error
on analytical functions and engineering emulators, and compared to Monte
Carlo, LHS, Halton, Sobol' and a rank-1 lattice. LAT is better than Monte
Carlo and LHS as soon as the integrand has interactions. Sobol' remains
more accurate at the powers of two, but its error is one to two orders of
magnitude larger at other sample sizes. The accuracy of LAT does not
depend on the sample size. Finally, the cells form a partition of the
hypercube for any $N$. This allows one to refine the design locally, to
search for an optimum by splitting cells and to sample non-rectangular
regions.
\end{abstract}


\section{Introduction}
\label{sec:intro}

One of the main objectives when performing costly numerical experiments
is to learn the behaviour of a quantity of interest with as few runs as
possible. The number of samples is noted $N$ and the dimension $d$. As
$N$ is often small for the dimension $d$, the Design of Experiments (DoE)
should fill the space as evenly as possible
\citep{sacks1989,johnson1990,fang2006,joseph2016}. Two families of
methods are classically used. Latin Hypercube Sampling (LHS) guarantees
uniform one-dimensional margins \citep{mckay1979}, and low discrepancy
sequences such as Halton and Sobol' control the joint distribution of the
points \citep{halton1960,sobol1967,niederreiter1992,dickeltz2010}. Both
are readily available, for instance in \texttt{scipy.stats.qmc}
\citep{roy2023}. Each has a limitation in practice. LHS only controls the
margins. A digital net in base two is balanced only when $N$ is a power
of two, whereas the number of runs one can afford is set by the cost of
the model.

Aperiodic tilings offer another way to fill the space. The Penrose and
Ammann--Beenker tilings, and more generally the cut-and-project sets they
come from, fill the plane with a few prototiles and never repeat,
although they present long-range order
\citep{baakegrimm2013,senechal1995,debruijn1981}. Placing one point per
tile is then a natural way to build a design, and the tile areas are
quadrature weights. This idea has been used in computer graphics
\citep{ostromoukhov2004,kopf2006,ahmed2016}, in the plane only. Cutting
the space into cells of equal volume for an arbitrary $N$, and combining
such a stratification with Latin margins, both have antecedents.
Recursive stratified integration splits boxes adaptively
\citep{press1990}, \citet{saka2007} Latinize centroidal Voronoi point sets
and \citet{shields2016} impose Latin margins on user-chosen strata.
\citet{he2016} stratify the unit interval along a space-filling curve and
\citet{clement2024} build equal-volume partitions from slabs. The closest
antecedent is the generalized stratified sampling of \citet{wessing2017}.
It cuts the longest edge of each box with a balanced allocation of the
points and Latinizes the result by a bipartite matching. The golden ratio
already generates point sets, such as the Fibonacci lattice
\citep{sloanjoe1994}, the sequence $R_d$ \citep{roberts2018} and nets in
base $\varphi$ \citep{kirk2025}, but no partition. To the best of our
knowledge, these Latinized stratifications have not been compared with
the scrambled sequences and a lattice rule at an arbitrary number of
samples, and their cells have not been used to refine a design, to
search for an optimum or to sample a non-rectangular region.

This work proposes a space-filling design of this family, named LAT for
\emph{Latinized aperiodic tiling}. This is a three-step process: (i) the
unit hypercube $[0,1]^d$ is recursively cut across its longest edge into
cells of equal volume following the golden section, (ii) one point is
placed in each cell, and (iii) the margins are Latinized by matching the
cells to the Latin bins, every point staying inside its own cell. The
construction is short to implement. It only uses integer splits, and the
cells can be built independently of each other. LAT differs from
\citet{wessing2017} by the golden cut, whose recursion has an exact
integer description at every sample size. It is also shown that Latin
margins and a point uniform in its cell are compatible, so that such a
design can be made unbiased. The cells of LAT are a partition of the
hypercube at any $N$, and the applications rely on this property. A
design can be refined locally while the rest of the sample is left
untouched, the cells can be split to search for an optimum, and a
non-rectangular region can be sampled by keeping the cells with their
centre inside it. These operations are usually handled by sequential
designs \citep{jones1998,crombecq2011,garud2017,liu2018,fuhg2021,roy2020}
and by designs optimized on the region
\citep{lekivetz2015,draguljic2012}.

The paper is organized as follows. Sect.~\ref{sec:criteria} presents the
quality criteria. Sect.~\ref{sec:construction} presents the construction
and its aperiodicity, and Sect.~\ref{sec:random} its randomization.
Sects.~\ref{sec:discrepancy} and~\ref{sec:integration} assess the
discrepancy and the integration error. Sect.~\ref{sec:sequential}
presents the local refinement, the global optimization and the
non-rectangular regions. Finally, conclusions and perspectives are drawn
in Sect.~\ref{sec:conclusion}.

\section{Quality criteria}
\label{sec:criteria}

Two families of criteria are used throughout this work: four
$L_2$-discrepancies, and the integration error on test functions with a
known integral. The centred ($\CD$), wrap-around ($\WD$) and mixture
($\MD$) discrepancies \citep{hickernell1998,zhou2013} measure how far the
empirical distribution of the points is from the uniform one. For a
design $X = \{x_i\}_{i=1}^{N} \subset [0,1]^d$ with coordinates
$x_{ij}$, writing $b_{ij} = \lvert x_{ij} - \tfrac12 \rvert$ and
$c_{ijk} = \lvert x_{ij} - x_{kj}\rvert$, where $i, k$ run over the $N$
points and $j$ over the $d$ coordinates,
\begin{align}
\CD^2(X) &= \Bigl(\tfrac{13}{12}\Bigr)^{\!d}
- \frac{2}{N}\sum_{i}\prod_{j}
  \Bigl(1 + \tfrac12 b_{ij} - \tfrac12 b_{ij}^{2}\Bigr)
+ \frac{1}{N^{2}}\sum_{i,k}\prod_{j}
  \Bigl(1 + \tfrac12 b_{ij} + \tfrac12 b_{kj}
        - \tfrac12 c_{ijk}\Bigr),
\label{eq:cd}\\
\WD^2(X) &= -\Bigl(\tfrac{4}{3}\Bigr)^{\!d}
+ \frac{1}{N^{2}}\sum_{i,k}\prod_{j}
  \Bigl(\tfrac32 - c_{ijk}\,(1 - c_{ijk})\Bigr),
\label{eq:wd}\\
\MD^2(X) &= \Bigl(\tfrac{19}{12}\Bigr)^{\!d}
- \frac{2}{N}\sum_{i}\prod_{j}
  \Bigl(\tfrac53 - \tfrac14 b_{ij} - \tfrac14 b_{ij}^{2}\Bigr)
+ \frac{1}{N^{2}}\sum_{i,k}\prod_{j}
  \Bigl(\tfrac{15}{8} - \tfrac14 b_{ij} - \tfrac14 b_{kj}
        - \tfrac34 c_{ijk} + \tfrac12 c_{ijk}^{2}\Bigr).
\label{eq:md}
\end{align}
The wrap-around discrepancy is invariant to a shift of each coordinate
modulo one, and the mixture discrepancy corrects some known weaknesses of
the two others \citep{zhou2013}. The fourth criterion is the average
squared discrepancy ($\ASD$) of \citet{clement2025}.
\citet{matousek1998} noted that the $L_2$ star discrepancy favours the
origin, to the point that $N$ copies of the vertex $(1, \ldots, 1)$ can
have a smaller value than $N$ independent uniform points. The $\ASD$
averages the squared $L_2$ star discrepancy over the $2^d$ vertices at
which it can be anchored. The average acts coordinate by coordinate and
gives
\begin{equation}
\ASD^2(X) = \Bigl(\tfrac13\Bigr)^{\!d}
- \frac{2}{N}\sum_{i}\prod_{j} \frac{1 + 2x_{ij}(1 - x_{ij})}{4}
+ \frac{1}{N^{2}}\sum_{i,k}\prod_{j} \frac{1 - c_{ijk}}{2}
\label{eq:asd}
\end{equation}
in $O(dN^2)$ operations. It is $4^{-d}$ times the symmetric discrepancy
of \citet{hickernell1998} with a weight of $4$ in every coordinate
\citep{clement2025}. The centred discrepancy has a similar weakness in
high dimension. $N$ copies of the centre of the cube give
$\CD^2 = (13/12)^d - 1$, whereas $N$ independent uniform points give
$((5/4)^d - (13/12)^d)/N$ in expectation. At $d = 30$ the copies are
thus the better design for this criterion whenever $N \le 79$. The first
three criteria are computed with \texttt{scipy.stats.qmc}
\citep{roy2023}. The values reported are the mean of the root
discrepancy over independent randomizations. The integration error is the
root-mean-square error (RMSE) of the sample mean on the test functions of
Sect.~\ref{sec:funcs}.

In the tables, each cell is coloured according to its ratio to LAT on the
same case, in blue when the method is better than LAT and in red when it
is worse. A case is a single row, or a single column when the methods are
the rows. The colour is proportional to the logarithm of the ratio, LAT
being left uncoloured, and it compares the methods within a case and not
across cases.

\section{Presentation of the method}
\label{sec:construction}

\subsection{Classical aperiodic tilings}
\label{sec:wall}

The most natural way to build a design from an aperiodic tiling is to
use the classical tilings themselves. Fig.~\ref{fig:prototiles} shows the
pentagrid Penrose tiling, the Ammann--Beenker tiling and a codimension-one
cut-and-project tiling. Each has one jittered point per tile, is cropped
to the unit square and is randomized through the translation of its
acceptance window \citep{baakegrimm2013}. Without Latinization, Penrose
and Ammann--Beenker are at the level of LHS in terms of discrepancy, with
slopes close to $-0.67$. In their standard orientation the tile centroids
also lie on a few hundred lines per coordinate, and on the type A
function of Sect.~\ref{sec:funcs} Penrose only reaches the Monte Carlo
rate. Once their margins are Latinized, the three tilings come within a
factor $1.3$--$1.5$ of LAT in terms of centred discrepancy. Hence, the
Latin margins do most of the work. Two difficulties remain. The number of
points kept by the crop fluctuates from one realization to the other, so
the construction is not defined at a given sample size, and the tiles cut
by the hypercube lose the constant volume of their prototile. In higher
dimension the tiling has to be grown in an ambient space and then
cropped, and $99\%$ of the points are discarded at $d = 5$. The golden
partition avoids all of it, as it tiles the hypercube itself in any
dimension.

\begin{figure}[tbp]
\centering
\includegraphics[width=0.62\textwidth]{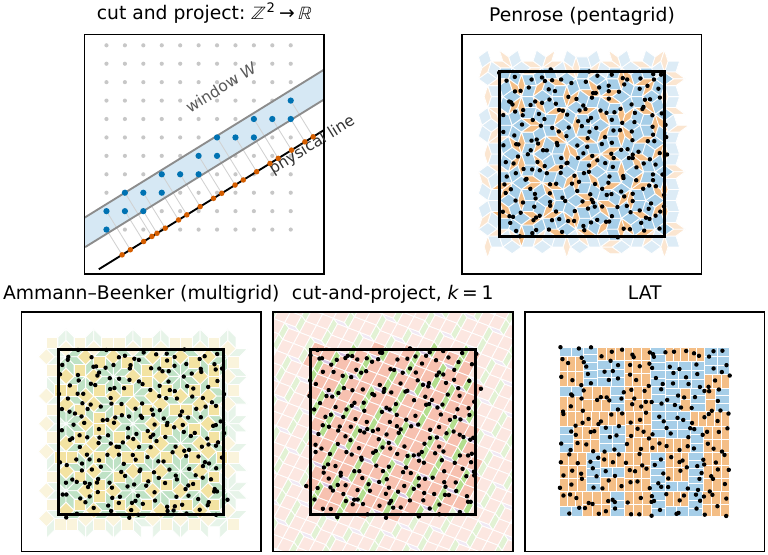}
\caption{The classical tilings. Top left, the cut-and-project mechanism:
the points of $\mathbb{Z}^2$ whose internal projection falls in the window
$W$ are projected onto the physical line, and give the Fibonacci chain.
An internal dimension $k$ gives Penrose ($k = 3$) and Ammann--Beenker
($k = 2$). One jittered point per tile, $N = 250$ points, the tiles
outside the crop box dimmed. LAT tiles the square itself.}
\label{fig:prototiles}
\end{figure}

\subsection{A golden-section tiling of the cube}
\label{sec:golden}

The general idea is to partition the unit hypercube recursively
(Algorithm~\ref{alg:lat}). Fig.~\ref{fig:construction} shows the
construction round by round. A cell that has to hold $j$ points is cut
across its longest edge, at the fraction $j_1/j$ of that edge, into two
cells holding $j_1 = \mathrm{round}(j/\varphi)$ and $j - j_1$ points,
where $\varphi = (1+\sqrt5)/2$ is the golden ratio. The recursion stops
when each cell holds a single point. Each of the $N$ leaf cells then has
a volume of $1/N$, whatever $N$ and $d$. The aspect ratios remain
bounded. All the cut fractions lie in $[\tfrac13, \tfrac23]$, the extreme
values being reached for $j = 3$, $6$ and $9$. An induction on the tree
then shows that cutting the longest edge under this constraint keeps the
ratio between the longest and the shortest edges at most three. The
diameter of the cells is thus $O(N^{-1/d})$, a property used in
Sect.~\ref{sec:rate}.

\begin{figure}[tbp]
\centering
\includegraphics[width=0.74\textwidth]{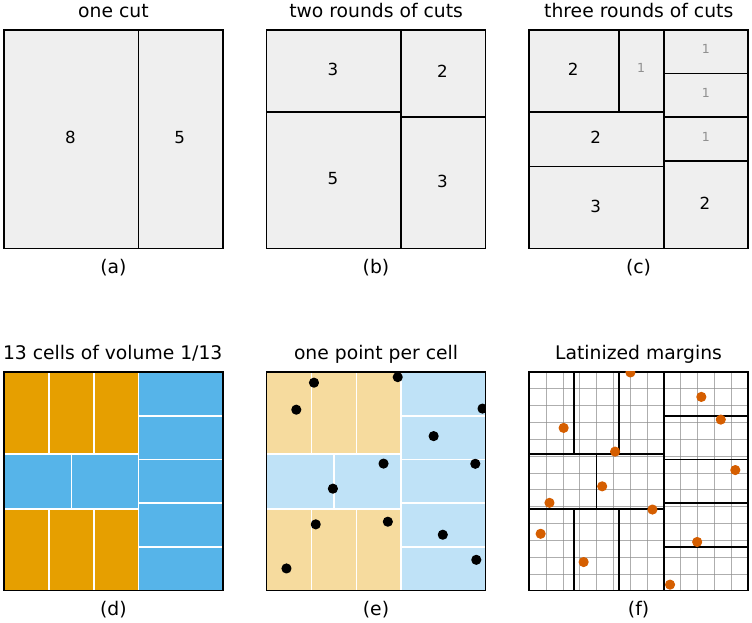}
\caption{The construction at $N = 13$. (a) The first cut gives $8$ and
$5$. (b) The same rule gives $5$ and $3$ from the $8$, and $3$ and $2$
from the $5$. (c) One round later, the cells that already hold a single
point are greyed. (d) After five rounds the $13$ cells have a volume of
$1/13$, coloured by orientation, wide in blue and tall in orange. (e) One
point is placed uniformly at random in each cell. (f) The Latinized
design: every point is inside its cell, and each row and each column of
the $13 \times 13$ grid holds exactly one point. The canonical partition
is drawn here, and its random exchanges are the subject of
Sect.~\ref{sec:random}.}
\label{fig:construction}
\end{figure}

\begin{algorithm}[tbp]
\caption{LAT: golden-section tiling of $[0,1]^d$ with one point per cell.}
\label{alg:lat}
\begin{algorithmic}[1]
\Require sample size $N$, dimension $d$
\State initialise a stack with the cell $\bigl([0,1]^d,\, N\bigr)$
\While{a cell $(C, j)$ on the stack has $j > 1$}
  \State $j_1 \gets \min\!\bigl(\max(\mathrm{round}(j/\varphi), 1),\, j-1\bigr)$
  \State cut $C$ across its longest edge at the fraction $j_1/j$
  \State replace $(C,j)$ by the two children $(C_1, j_1)$ and $(C_2, j-j_1)$
\EndWhile
\State match the cells to the $N$ Latin bins in each coordinate and place one point per cell inside its cell and its bins
\end{algorithmic}
\end{algorithm}

The construction is short to implement. It only involves integer splits
and one multiplication per cut, and it needs neither a table of
direction numbers, nor a generating vector, nor an optimization loop. The
tree has $2N - 1$ nodes and each cut costs $O(d)$, so the partition costs
$O(dN)$ operations and memory. The larger child holds about $N/\varphi$
points, so the depth of the tree is about $\log N / \log\varphi$, i.e.,
$18$ at $N = 8192$. Two properties follow. First, the cells of a level
do not depend on each other. A whole level is thus cut in a single
vectorized operation, and the loop only runs over the depth of the tree.
Second, the number of points of each child only depends on the integers.
The cell of a given index is thus found by a descent from the root in
$O(d \log N)$ operations, without building the tree, with the random
exchanges of Sect.~\ref{sec:random} drawn from a stream indexed by the
node. The cells can then be generated in parallel or on demand, as the
refinement of Sect.~\ref{sec:refine} does. The Latinization of
Sect.~\ref{sec:random} is independent from one coordinate to the other
and costs $O(N \log N)$ per coordinate, so the whole design costs
$O(dN \log N)$. In comparison, a scrambled Sobol' sequence costs $O(dN)$
but relies on tables of direction numbers, a lattice rule needs a
generating vector searched component by component \citep{nuyenscools2006},
and the optimized Latin hypercubes of Sect.~\ref{sec:discrepancy} are
iterative searches.

Since the cut fractions follow the golden section, the cells do not align
on a grid for most sample sizes. Some sizes are however degenerate. In
$d = 2$, the recursion gives an exact regular grid for each $N = F_m^2$
($4$, $9$, $25$, $64$, $169$, \ldots), with $F_m$ the Fibonacci numbers,
as well as for $N = 2$ and $6$. Indeed, the proof of Lemma~\ref{lem:round} below gives
$F_m^2/\varphi - F_m F_{m-1} = F_m \psi^{m-1}/\varphi$, of absolute value
$F_m \varphi^{-m} < \tfrac12$ for $m \ge 3$. Hence
$\mathrm{round}(F_m^2/\varphi) = F_m F_{m-1}$ and the recursion splits an $F_m \times F_m$ grid into two grids of the same
kind. An exhaustive search finds no other case below $4096$ in $d = 2$,
and only $N = 2$ and $4$ in $d = 3$, $4$ and $5$. For these sizes the
random draw of the tiling has no effect. $N = 64$ is the first size of
the integration studies below, and removing it moves the fitted slopes of
LAT by at most $0.03$ in $d = 2$.

\subsection{Aperiodicity}
\label{sec:aperiodic}

Let us now show that the split recursion is aperiodic. In the following,
$F_m$ denotes the Fibonacci numbers, with $F_1 = F_2 = 1$.

\begin{lemma}
\label{lem:round}
For all $m \ge 2$, $\;\lvert F_m/\varphi - F_{m-1}\rvert = \varphi^{-m} <
\tfrac12$, so $\mathrm{round}(F_m/\varphi) = F_{m-1}$.
\end{lemma}

\begin{proof}
With $\psi = -1/\varphi$ and $F_m = (\varphi^m - \psi^m)/\sqrt5$,
\[
\frac{F_m}{\varphi} - F_{m-1}
= \frac{\psi^{m-1}}{\sqrt5}\Bigl(1 - \frac{\psi}{\varphi}\Bigr)
= \frac{\psi^{m-1}}{\sqrt5}\,(3 - \varphi)
= \frac{\psi^{m-1}}{\varphi},
\]
since $(3-\varphi)\varphi = 2\varphi - 1 = \sqrt5$. Its absolute value is
$\varphi^{-m} \le \varphi^{-2} < \tfrac12$.
\end{proof}

By Lemma~\ref{lem:round}, when $N = F_m$ a cell holding $F_i$ points is
split into two cells holding $F_{i-1}$ and $F_{i-2}$ points. The
recursion is thus the Fibonacci tree, and its structure is the one of the
Fibonacci word \citep{lothaire2002}. The split sizes are recorded as a
word $W(j)$. The leaves of a tree of size $j$ are read from left to
right, with the letter $b$ for the singleton produced whenever a cell of
size three is cut and the letter $a$ otherwise.

\begin{proposition}
\label{prop:fib}
For $m \ge 5$, $\,W(F_m) = W(F_{m-1})\,W(F_{m-2})$, the number of letters $b$
in $W(F_m)$ is $F_{m-3}$, and the frequency of $b$ tends to $\varphi^{-3}$,
which is irrational. Thus, the recursion word is not eventually
periodic.
\end{proposition}

\begin{proof}
The concatenation is Lemma~\ref{lem:round} applied at the root. It makes
each $W(F_{m-1})$ a prefix of $W(F_m)$, and the infinite recursion word
is well defined as their limit. The count of $b$ obeys
$B_m = B_{m-1} + B_{m-2}$ with $B_3 = 0$ and $B_4 = 1$, so $B_m = F_{m-3}$
and $B_m/F_m \to \varphi^{-3}$. A word that is eventually periodic has a
rational letter frequency, and $\varphi^{-3}$ is irrational.
\end{proof}

Proposition~\ref{prop:fib} concerns the infinite word. A design uses a
finite word of length $N$, and the relevant quantity is then its shortest period. This period is computed exactly. It is $F_{m-1}$ at every Fibonacci
size $N = F_m$ up to $6765$ and $\mathrm{round}(N/\varphi)$ at
$N = 1000$, $1024$ and $5000$. It is never below $N/\varphi$, so the
pattern of cell sizes does not repeat even once inside a design.

The golden ratio is not the only possible cut. The silver ratio
$1 + \sqrt2$, the plastic ratio $\rho \approx 1.325$, the balanced rule
$p = 1/2$ of \citet{wessing2017} and a control drawing the larger
fraction in $U(0.5, 0.7)$ at each cut were tested on 120 shared seeds,
with the greedy matching of Sect.~\ref{sec:random}. At $N = 1024$ the
five rules are within four percent of each other on both discrepancies
in $d = 2$, $5$ and $10$, and the balanced rule is the most accurate. The
balanced rule is however a regular grid at this size, and the cells of a
grid are aligned with the Latin bins. The greedy matching is then as good
as the exact one of Proposition~\ref{prop:unbiased}, within $2\%$,
whereas it is $3$ to $18\%$ worse on the golden cells. Given the exact
matching, the golden cells become the more accurate at $N = 1000$ and
$1024$, by $10$ to $13\%$ in $d = 2$, $7$ to $9\%$ in $d = 5$ and $1.5$
to $3\%$ in $d = 10$. The design of \citet{wessing2017} itself, with its
bipartite matching, is within $4\%$ of LAT on both discrepancies, and its
small advantage at $d = 5$ also comes from its matching. The cut rule is
thus of second order compared to the Latin margins. The golden section is
retained for the exact integer description of its recursion, $\varphi$
tying the split sizes to the Fibonacci numbers. It also avoids the grid
better than the balanced rule: the scan of Sect.~\ref{sec:golden} up to
$N = 4096$ in $d = 2$ returns ten degenerate sizes for the golden cut and
seventeen for the balanced one, among them all the powers of two.

\section{Randomization}
\label{sec:random}

The randomization of LAT takes place at three levels. First, the tiling
itself is drawn anew for each realization. At each cut the two children
are exchanged with a probability of one half, and ties for the longest
edge are broken at random. Each cell keeps a volume of $1/N$ and only the
layout changes. Second, the point of each cell is placed uniformly at
random inside it. This one-point-per-cell rule is the stratified
estimator of \citet{haber1966}. Drawing the tiling as well makes the
construction a randomized quasi-Monte Carlo rule in the tradition of
\citet{cranley1976} and \citet{owen1995}. At this level the estimate is
unbiased for any integrable function, and independent randomizations
estimate the error. Third, the margins are made Latin. In each coordinate
the $N$ cells are matched one-to-one with the $N$ bins $[k/N, (k+1)/N)$
they overlap, and the point of a cell is drawn uniformly in the
intersection of its cell with its bin. Each bin then holds exactly one
point, as in a Latin hypercube \citep{mckay1979}, and each point stays
inside its cell. A perfect matching always exists. Indeed, $k$ cells have
a union of volume $k/N$, whose projection on the coordinate has a length
of at least $k/N$ and meets at least $k$ bins. This is Hall's condition.
A first matching is built by sweeping the bins in order and giving each
one to the overlapping cell whose interval ends first, a greedy rule that
is optimal for intervals \citep{glover1967}. This matching is
structured, so it is randomized by swaps that keep it perfect: two cells
exchange their bins when each bin lies in both cells. The whole costs
$O(N \log N)$ per coordinate. A coordinate-wise rank transform, the
recipe of \citet{saka2007} and \citet{shields2016}, also gives Latin
margins but lets a point leave its cell. It is within $1.5\%$ of the
matching on the discrepancies at $N = 1024$, so keeping the point inside
its cell costs nothing. On the canonical tiling at $d = 2$ and
$N = 1024$, jittering the points divides the centred discrepancy of the
centroids by $2.7$, and the Latinization divides it by a further $2.5$.

Whether the third level preserves the unbiasedness depends on the law of
the matching. The point of a cell is uniform in the intersection of its
cell with its bin and not in the whole cell. Write $C_i$ for the
projection of cell $i$ on the coordinate at hand, $w_i$ for its length,
$B_j = [j/N, (j+1)/N)$ for the bins and
\begin{equation}
q_{ij} = |C_i \cap B_j| / w_i
\label{eq:claim}
\end{equation}
for the fraction of the cell covered by a bin. These are the frequencies
a uniform point in the cell would give.
\begin{proposition}
\label{prop:unbiased}
Let $[0,1]^d$ be partitioned into $N$ axis-aligned boxes of volume $1/N$.
In each coordinate the matrix $q$ of \eqref{eq:claim} is doubly
stochastic. There are therefore perfect matchings $\sigma_1, \ldots,
\sigma_m$ of the cells with the bins and weights $\lambda_t \ge 0$ of sum
one such that $q = \sum_t \lambda_t P_{\sigma_t}$, with $P_\sigma$ the
permutation matrix of $\sigma$. Draw the matching as $\sigma_t$ with
probability $\lambda_t$, independently in each coordinate, and the point
of a cell uniformly in the intersection of its cell with its bins. That
point is then uniform in its cell. The margins are exactly Latin, every
point lies inside its cell, and the sample mean is unbiased for any
integrable function.
\end{proposition}

\begin{proof}
The rows of $q$ sum to one, as the bins cover $C_i$. For the columns,
cell $i$ has volume $1/N = w_i A_i$ with $A_i$ the area of its face
orthogonal to the coordinate, so $1/w_i = N A_i$ and
$\sum_i q_{ij} = N \sum_i A_i |C_i \cap B_j|$. Each term is the volume of
the part of cell $i$ lying in the slab of the bin, and the cells
partition the cube, so the sum is $N$ times the volume of that slab,
which is one. A doubly stochastic matrix is a convex combination of
permutation matrices \citep{birkhoff1946}. This gives the $\lambda_t$.
Under this law cell $i$ holds bin $j$ with probability $q_{ij}$ and its
point is uniform on $C_i \cap B_j$, so the density of the coordinate on
$C_i$ is $q_{ij} / |C_i \cap B_j| = 1/w_i$. The coordinates are matched
independently, hence the point is uniform in the cell. Every cell has
volume $1/N$ and holds one uniform point, so the expectation of the
sample mean is the sum of the integrals of $f$ over the cells.
\end{proof}

The matching of Proposition~\ref{prop:unbiased} is drawn without listing
the $\sigma_t$, by dependent rounding \citep{gandhi2006}. A cycle of the
graph of the fractional entries is walked and the entries are moved along
it alternately, by an amplitude drawn between the two values that keep
them in $[0, 1]$. Each step leaves the row and column sums and every
expectation unchanged, and makes one more entry integral. The number of
steps is the number of cell-to-bin incidences. A cell of volume $1/N$
with a bounded aspect ratio meets $O(N^{1-1/d})$ bins in each
coordinate, so this rule costs $O(N^{2-1/d})$ per coordinate against the
$O(N \log N)$ of the greedy one. The studies below therefore use the
greedy rule, and the price of that choice is measured. In the smallest
case, $d = 2$ and $N = 3$, the point of a cell lands in the central cell
of the $3\times3$ bin grid with probability $0.286$ instead of $1/3$
under the greedy rule, and with probability $0.332$ under the rule of the
proposition, over $2\times10^5$ randomizations. On the type A and C
integrands, over $2\times10^4$ randomizations per case with
$d \in \{2, 5\}$ and $N \in \{64, 256\}$, the greedy rule departs from
zero by more than two standard errors in one case out of eight, with
$t = 3.7$, and its bias is at most $2.6\%$ of its RMSE. The rule of the
proposition stays below $1.6$ standard errors and $1.2\%$. Over 16 shared
tilings the exact rule is also better in terms of centred discrepancy,
by $1$ to $16\%$ at $N = 256$ and $1024$ in $d = 2$, $5$ and $10$. The
greedy matching is thus slightly biased and slightly less well spread.
The refinement estimator of Sect.~\ref{sec:refine} uses the jittered rule
without Latinization, so that its unbiasedness is the proved case.

\section{Space-filling quality}
\label{sec:discrepancy}

Let's first take a look at Fig.~\ref{fig:convergence}. It compares LAT
with crude Monte Carlo (MC), LHS, and the scrambled Halton
\citep{owen2017halton} and Sobol' sequences over 16 randomizations. The
Sobol' points are scrambled by the linear matrix scrambling of
\citet{matousek1998} followed by a digital shift, as implemented in
\texttt{scipy.stats.qmc}. The four discrepancies give the same order up
to $d = 10$: Sobol', then Halton, LAT, LHS and finally MC. The WD and MD
are not shown as they give the same picture as the CD. LAT is better than
MC in all dimensions, as any equal-volume stratification is in terms of
expected squared discrepancy \citep{kiderlen2022}. In low dimension it
converges faster than LHS, with a slope of $-0.82$ compared to $-0.54$ on
the CD at $d = 2$. This is due to the fact that LHS only improves the
margins and keeps the Monte Carlo rate on the joint distribution. LAT
remains behind the two sequences, by a factor of up to five for
$d \le 10$, for the reason given in Sect.~\ref{sec:rate}. As the
dimension grows, all these differences are compressed. At $d = 30$ all
the methods converge at the Monte Carlo rate, and the structured designs
only keep a small advantage on the constant. The ASD does not show this
advantage. At $d = 30$ its square is dominated by the diagonal terms of
the double sum in \eqref{eq:asd}. These terms give $2^{-d}/N$ whatever the design, and the five designs agree within $0.05\%$ at every $N$.
\citet{clement2025} note that a high dimension calls for weights on the
coordinates, and this is the reason.

\begin{figure}[tbp]
\centering
\includegraphics[width=0.74\textwidth]{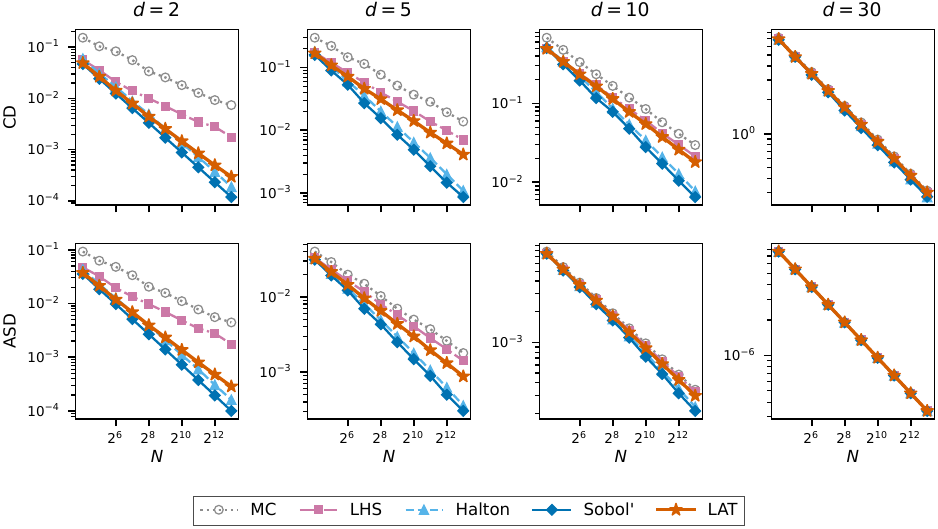}
\caption{Convergence of the centred discrepancy (top) and of the average
squared discrepancy (bottom) with respect to the sample size in
$d = 2, 5, 10, 30$, 16 randomizations on the same realizations.}
\label{fig:convergence}
\end{figure}

Plain LHS is the weakest member of its family. Its optimized variants are
also considered: the CD-optimized hypercube of \texttt{scipy.stats.qmc}
and the maximin hypercube of \citet{morris1995}. The latter minimizes
$\varphi_p = (\sum_{i<k} \lVert x_i - x_k \rVert_2^{-p})^{1/p}$ with
$p = 50$. As expected, the CD-optimized hypercube is the best of the
Latin family. At $N = 1024$ it is better than LAT in all dimensions, by
a factor between $1.1$ and $1.7$. At $N = 8192$ these factors shrink, and
the order is even reversed at $d = 2$ since the slope of LAT is steeper.
The optimization thus improves the constant and not the rate. It also
has a cost. Every candidate swap costs $O(dN)$ and the number of swaps
$T$ is a search budget, so an optimized design costs $O(TdN)$ where LAT
is built in one pass. If a design is built once and each simulation is
expensive, this cost does not matter and the optimized design is the
right choice. If the design has to be rebuilt many times, in a
cross-validation loop, a replication study or an adaptive scheme, a
one-shot construction scales better.

\section{Integration}
\label{sec:integration}

\subsection{Test functions and protocol}
\label{sec:funcs}

The core of the study is the type A, B and C benchmark of
\citet{kucherenko2015}, with the same protocol as in \citet{roy2020}:
\begin{equation}
f_{\mathrm{A}}(x) = \prod_{i=1}^{d}
   \frac{\lvert 4x_i - 2\rvert + a_i}{1 + a_i}, \quad a_i = i;
\qquad
f_{\mathrm{B}}(x) = \prod_{i=1}^{d} \frac{d - x_i}{d - 1/2};
\qquad
f_{\mathrm{C}}(x) = 2^{d} \prod_{i=1}^{d} x_i.
\label{eq:abc}
\end{equation}
All three integrate to one. Type~A has a few important variables since
the weight of $x_i$ decreases with $a_i$. Type~B has all its variables
equally important with weak interactions. Type~C has all its variables
important with strong interactions. A geometric integrand is added,
\begin{equation}
f_{\mathrm{G}}(x) = \mathbf{1}\{\langle x, u\rangle \le t\},
\qquad t = \tfrac12 \textstyle\sum_i u_i,
\label{eq:geom}
\end{equation}
where $u$ is a unit vector with positive components, drawn anew for each
randomization by normalizing $|z| + \tfrac14$ with $z$ standard normal.
Its exact value is $\tfrac12$ whatever $u$ and $d$, and it is
discontinuous across a plane that is not aligned with the axes. Four
synthetic examples of \citet{owen2022drop} complete the test,
\begin{equation}
f_1 = \prod_{i=1}^{2} \Bigl(x_i - \tfrac12\Bigr), \quad
f_2 = \sum_{i=1}^{5} \bigl(e^{x_i} + 1 - e\bigr), \quad
f_4 = \prod_{i=1}^{3} \bigl(e^{x_i} + 1 - e\bigr), \quad
f_6 = \sum_{i=1}^{d} x_i,
\label{eq:morefunc}
\end{equation}
with exact integrals $0$, $0$, $0$ and $d/2$. The pure interaction $f_1$
has no main effect, $f_2$ and $f_6$ are additive, and $f_6$ is run for
$d = 5$ and $30$. Finally, Sect.~\ref{sec:emulator} uses five
engineering emulators: the borehole, OTL circuit, piston, wing weight and
robot arm functions, with inputs uniformly distributed over the standard
ranges of the virtual library of \citet{surjanovic2013}.

Comparisons of this kind are classically run for $N = 2^k$, here from
$2^6$ to $2^{13}$. These are the only sizes recommended by
\citet{owen2022drop} for Sobol' points. However, the number of
simulations one can afford is rarely a power of two, $1000$ runs for
instance. The study is thus also carried out on a dense grid of $58$
sizes, eight per octave with the powers of two and $N = 1000$ included,
on type A at $d = 2$ and $5$, type B at $d = 5$, $f_2$ and $f_4$. Every
reported error is a RMSE over $99$ independent randomizations, each LAT
realization drawing a new tiling.

One caveat applies at $d = 2$. The scrambling of \citet{matousek1998}
integrates exactly a product of coordinate functions that are linear on
each dyadic half, and $86$ of the $99$ randomizations of Sobol' return
the exact value of the type A integrand at $N = 8192$. The RMSE then
comes from a few randomizations. Each case is thus screened by its
effective sample size, $\mathrm{ESS} = (\sum_r e_r^2)^2 / \sum_r e_r^4$
with $e_r$ the error of randomization $r$, and a value is marked with a
dagger when $\mathrm{ESS} < \max(10, R/20)$. The type A, B and C cells of
Sobol' at $d = 2$ are recomputed with $4000$ randomizations. With $R$
randomizations the relative standard error of a RMSE is
$\tfrac12\sqrt{1/\mathrm{ESS} - 1/R}$. Two methods are said to be
separated when their errors differ by more than twice the combined
standard error, and the leading group of a comparison is the most
accurate method together with those not separated from it.

\subsection{Results at the powers of two}
\label{sec:results}

Let's first take a look at Fig.~\ref{fig:integration}. The results are
in accordance with the ones obtained on the discrepancy, and two factors
govern them: the class of the integrand and the dimension. On an
additive integrand only the margins matter. LAT and LHS then have the
same error on $f_6$, and both converge close to $N^{-3/2}$.
\citet{stein1987} proved that a Latin hypercube removes the additive part
of the variance. The margins of LAT are Latin although their law is not
exactly the one of a Latin hypercube, and the measured errors of the two
designs are the same on every additive case here. Scrambled Halton has no
Latin structure and is worse by a factor $210$ at $d = 5$ and $420$ at
$d = 30$ on $f_6$. As soon as the integrand has interactions, LAT is
better than LHS. The gain grows with the order of the interactions and
reaches a factor sixty on $f_1$, a function without main effect. On smooth
products the scrambled nets are far better than the stratified designs,
since they reach a variance of order $N^{-3}(\log N)^{d-1}$
\citep{owen1997} whereas one uniform point per cell of an equal-volume
partition gives $O(N^{-1-2/d})$. On the discontinuous $f_{\mathrm{G}}$
this advantage disappears. The nets fall back to the stratification rate
\citep{hewang2015}, and the gap with LAT is below a factor $1.35$.

\begin{figure}[tbp]
\centering
\includegraphics[width=0.70\textwidth]{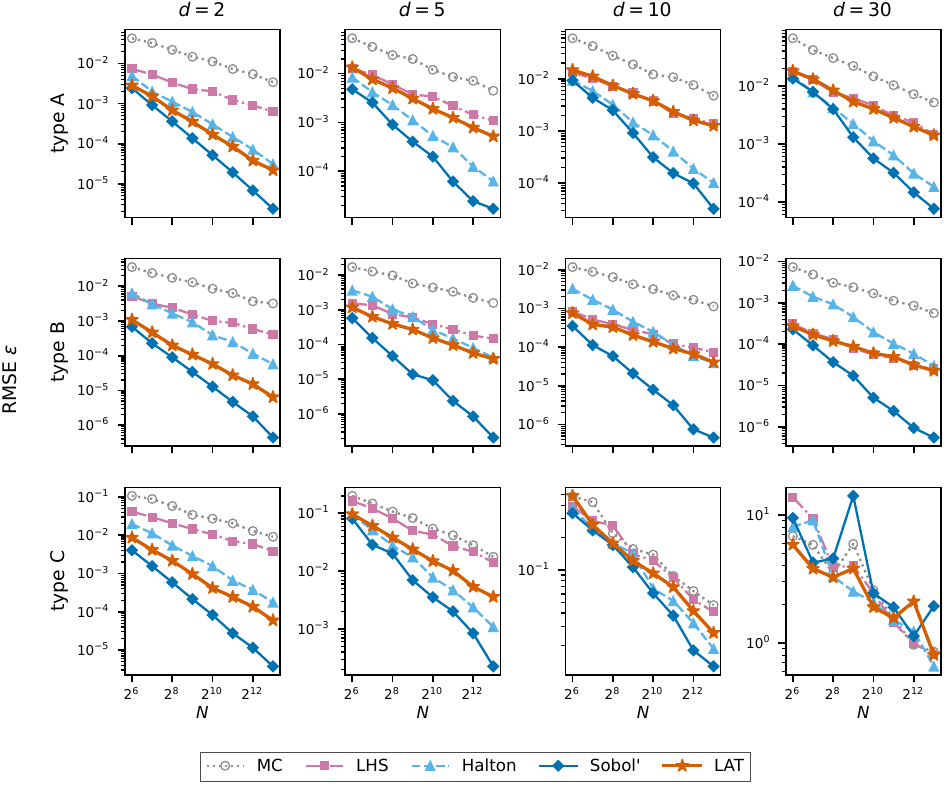}
\caption{Integration error with respect to the sample size on the
type~A, B and C benchmark, $d = 2, 5, 10, 30$, 99 randomizations
($2^{14}$ for type~C at $d = 30$, $4000$ for the type A, B and C cells of
Sobol' at $d = 2$). Panels are scaled independently.}
\label{fig:integration}
\end{figure}

The dimension compresses all the differences, since the stratification
rate $N^{-1/2-1/d}$ of a smooth integrand tends to $N^{-1/2}$. At
$d = 30$ LAT is still better than MC. It is level with LHS on type~A and
better than LHS by $9\%$ on type~B, and its slope is back to the Monte
Carlo value. Type~C at $d = 30$ has a variance of $(4/3)^{d} - 1$ with a
heavy right tail, and no method resolves it at this budget.

\subsection{Arbitrary sample sizes}
\label{sec:anyn}

Moving on to the sample sizes between the powers of two,
Fig.~\ref{fig:anyn} presents the dense sweep. Using Sobol' at another
size means taking the first $N$ points of the scrambled sequence, and
\texttt{scipy.stats.qmc} warns against it \citep{roy2023}. The saw-tooth
pattern is the one of this truncated net. Its error is the largest just
after each power of two and then decreases until the next one
\citep{owen2022drop}. LAT follows a power law across the same sizes: the
median variation inside an octave matches the decay per octave within
fifteen percent in each case. It keeps the rates obtained at the powers
of two, $-1.02$ on type A at $d = 2$ and $-1.5$ on $f_2$. On these two
cases LAT is close to the net at the powers of two, and the truncated net
is worse than LAT at each other size, with one exception on type A. The
median factor is $2.7$ and $16$ resp.

\begin{figure}[tbp]
\centering
\includegraphics[width=0.78\textwidth]{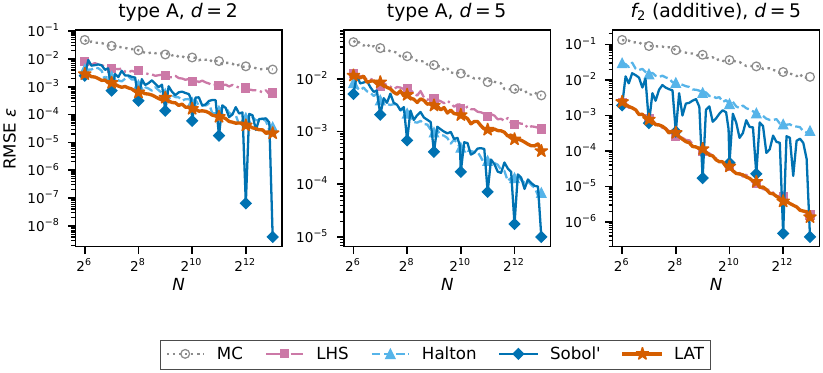}
\caption{Integration error on a dense grid of sample sizes, eight per
octave with the powers of two (markers) and $N = 1000$ included, 99
randomizations. The saw-tooth pattern is Sobol' truncated to the first
$N$ points of the scrambled sequence. Slopes over all $58$ sizes are in
Table~\ref{tab:anyn}. Panels are scaled independently.}
\label{fig:anyn}
\end{figure}

\begin{table}[tbp]
\centering
\scriptsize
\setlength{\tabcolsep}{1.5pt}
\begin{tabular}{lcccccc}
\toprule
 & MC & LHS & Halton & Sobol' & LAT & Sobol' $1024$ \\
\midrule
type A, $d = 2$ & \cellcolor{cmpworse!40}$1.1\!\times\!10^{-2}$ ($-0.49$) & \cellcolor{cmpworse!21}$1.6\!\times\!10^{-3}$ ($-0.51$) & \cellcolor{cmpworse!6}$2.9\!\times\!10^{-4}$ ($-1.01$) & \cellcolor{cmpworse!7}$3.6\!\times\!10^{-4}$ ($-1.33^{\dagger}$) & $1.7\!\times\!10^{-4}$ ($-1.01$) & $6.0\!\times\!10^{-5}$$^{\dagger}$ \\
type A, $d = 5$ & \cellcolor{cmpworse!40}$1.3\!\times\!10^{-2}$ ($-0.50$) & \cellcolor{cmpworse!11}$3.3\!\times\!10^{-3}$ ($-0.50$) & \cellcolor{cmpbetter!40}$5.0\!\times\!10^{-4}$ ($-0.98$) & \cellcolor{cmpbetter!38}$5.4\!\times\!10^{-4}$ ($-1.06^{\dagger}$) & $2.0\!\times\!10^{-3}$ ($-0.67$) & $1.7\!\times\!10^{-4}$ \\
type B, $d = 5$ & \cellcolor{cmpworse!40}$4.7\!\times\!10^{-3}$ ($-0.51$) & \cellcolor{cmpworse!10}$3.6\!\times\!10^{-4}$ ($-0.50$) & \cellcolor{cmpworse!7}$2.7\!\times\!10^{-4}$ ($-0.94$) & \cellcolor{cmpbetter!40}$7.4\!\times\!10^{-5}$ ($-1.12^{\dagger}$) & $1.5\!\times\!10^{-4}$ ($-0.69$) & $9.4\!\times\!10^{-6}$$^{\dagger}$ \\
$f_2$, $d = 5$ & \cellcolor{cmpworse!40}$3.4\!\times\!10^{-2}$ ($-0.51$) & \cellcolor{cmpbetter!40}$3.7\!\times\!10^{-5}$ ($-1.50$) & \cellcolor{cmpworse!24}$2.2\!\times\!10^{-3}$ ($-0.94$) & \cellcolor{cmpworse!14}$4.1\!\times\!10^{-4}$ ($-1.20^{\dagger}$) & $4.0\!\times\!10^{-5}$ ($-1.52$) & $4.6\!\times\!10^{-5}$$^{\dagger}$ \\
$f_4$, $d = 3$ & \cellcolor{cmpworse!40}$4.0\!\times\!10^{-3}$ ($-0.50$) & \cellcolor{cmpworse!39}$3.9\!\times\!10^{-3}$ ($-0.48$) & \cellcolor{cmpbetter!6}$7.2\!\times\!10^{-4}$ ($-0.99$) & \cellcolor{cmpbetter!40}$5.4\!\times\!10^{-4}$ ($-1.04^{\dagger}$) & $7.5\!\times\!10^{-4}$ ($-0.83$) & $1.6\!\times\!10^{-4}$ \\
\bottomrule

\end{tabular}
\caption{Integration error at $N = 1000$ with the slope fitted over all
$58$ sizes of the sweep, 99 randomizations. Sobol' at $N = 1000$ takes
the first $1000$ points of the scrambled sequence, and the last column
gives the same sequence at its own size $N = 1024$. A daggered slope
relies on sizes with an effective sample size below $10$.}
\label{tab:anyn}
\end{table}

Table~\ref{tab:anyn} confirms this at $N = 1000$. Going from $1024$ to
$1000$ points multiplies the error of the net by a factor of three. The
leading group at $N = 1000$ is LAT alone on type A at $d = 2$, and LAT
together with LHS on $f_2$. It is the truncated net alone on type B and
on $f_4$, and the truncated net together with scrambled Halton on type A
at $d = 5$. The sequences thus do not fail away from the powers of two,
but their best accuracy requires them. LAT itself is not affected by the
sample size.

Three other designs are defined for any $N$. The stratification of
\citet{he2016} cuts the unit interval into $N$ equal segments, draws one
point in each and maps them through a $d$-dimensional Hilbert curve. Its
cells have a volume of $1/N$ but its margins are not Latin, and it is
worse than LAT by a factor $3$ to $9$ on types A and B. After a
coordinate-wise rank transform, it is within $2.5\%$ of LAT on every
discrepancy and within $20\%$ on the integration errors. Hence, the
accuracy of LAT comes from the combination of an equal-volume
stratification with exact margins, and aperiodicity is not by itself a
source of accuracy. The sequence $R_d$ of \citet{roberts2018} and a
randomly shifted rank-1 lattice \citep{sloanjoe1994,dickkuosloan2013}
have equidistributed margins at any $N$. They are far better than LAT on
the periodic type A family. On smooth non-periodic integrands an
equidistributed grid keeps the boundary term of the rectangle rule, and
the lattice is $30$ times worse than LAT on $f_2$ at $N = 1000$. A lattice
rule is thus the right choice when the integrand is periodic or nearly
so, and LAT when it is not.

\subsection{Engineering cases}
\label{sec:emulator}

The comparison is now repeated on the five emulators, following the
numerical example of \citet{owen2026}. Table~\ref{tab:emulator} reports
the variance reduction with respect to Monte Carlo at equal cost, i.e.,
the Monte Carlo variance $s^2/(NR)$ divided by the variance $S^2/R$ of the
method over $R = 200$ randomized estimates, with $s^2$ estimated from
$2\times10^5$ Monte Carlo points. Its relative standard error is
$\sqrt{1/\mathrm{ESS} - 1/R}$, about $12\%$ here, so that two ratios are
resolved when they differ by a factor $1.4$. These emulators are smooth
functions of low effective dimension, and the hierarchy is the one
expected for this class. The scrambled net is alone in the leading group
of each of them, by several orders of magnitude. Its effective sample
size however falls to $4.8$ at the dyadic sizes, and these cells are
daggered. LAT reduces the Monte Carlo variance by a factor ranging from
three to a thousand, and it is better than LHS on the five functions at
$N = 2^{13}$. The exception is the robot arm. It depends jointly on all its inputs and all the methods gain little on it
\citep{caflisch1997}. Truncating the net to $N = 1000$ costs it one to
two orders of magnitude on the four functions where it gains anything,
and it remains the best method on four of the five functions. The ratios
of LAT differ by at most $16\%$ between $1024$ and $1000$, within the
resolution of the experiment.

\begin{table}[tbp]
\centering
\scriptsize
\setlength{\tabcolsep}{2.6pt}
\begin{tabular}{lcccccccccc}
\toprule
 & & \multicolumn{4}{c}{$N = 2^{13}$} & \multicolumn{4}{c}{$N = 1000$}
 & $N{=}1024$ \\
\cmidrule(lr){3-6}\cmidrule(lr){7-10}\cmidrule(lr){11-11}
function & $d$ & LHS & Halton & Sobol' & LAT & LHS & Halton & Sobol' & LAT
 & Sobol' \\
\midrule
borehole & 8 & \cellcolor{cmpworse!40}$21.7$ & \cellcolor{cmpbetter!17}$3.3\!\times\!10^{3}$ & \cellcolor{cmpbetter!40}$2.9\!\times\!10^{5}{}^{\dagger}$ & $1.0\!\times\!10^{2}$ & \cellcolor{cmpworse!22}$30.9$ & \cellcolor{cmpbetter!7}$2.8\!\times\!10^{2}$ & \cellcolor{cmpbetter!15}$1.3\!\times\!10^{3}$ & $72.8$ & $4.1\!\times\!10^{4}$ \\
OTL circuit & 6 & \cellcolor{cmpworse!40}$1.3\!\times\!10^{2}$ & \cellcolor{cmpbetter!6}$2.5\!\times\!10^{3}$ & \cellcolor{cmpbetter!40}$4.6\!\times\!10^{7}{}^{\dagger}$ & $1.1\!\times\!10^{3}$ & \cellcolor{cmpworse!19}$1.6\!\times\!10^{2}$ & \cellcolor{cmpworse!6}$3.2\!\times\!10^{2}$ & \cellcolor{cmpbetter!8}$3.7\!\times\!10^{3}$ & $4.4\!\times\!10^{2}$ & $8.2\!\times\!10^{5}{}^{\dagger}$ \\
piston & 7 & \cellcolor{cmpworse!40}$15.1$ & \cellcolor{cmpbetter!18}$1.3\!\times\!10^{3}$ & \cellcolor{cmpbetter!40}$4.5\!\times\!10^{4}$ & $69.2$ & \cellcolor{cmpworse!18}$20.0$ & \cellcolor{cmpbetter!11}$2.5\!\times\!10^{2}$ & \cellcolor{cmpbetter!18}$7.1\!\times\!10^{2}$ & $39.4$ & $6.4\!\times\!10^{3}$ \\
wing weight & 10 & \cellcolor{cmpworse!40}$85.8$ & \cellcolor{cmpbetter!6}$8.9\!\times\!10^{2}$ & \cellcolor{cmpbetter!40}$1.1\!\times\!10^{6}{}^{\dagger}$ & $2.3\!\times\!10^{2}$ & \cellcolor{cmpworse!26}$84.1$ & \cellcolor{cmpworse!11}$1.2\!\times\!10^{2}$ & \cellcolor{cmpbetter!13}$2.2\!\times\!10^{3}$ & $1.6\!\times\!10^{2}$ & $5.8\!\times\!10^{4}$ \\
robot arm & 8 & \cellcolor{cmpworse!40}$1.6$ & \cellcolor{cmpbetter!31}$20.8$ & \cellcolor{cmpbetter!40}$37.2$ & $3.0$ & \cellcolor{cmpworse!8}$1.9$ & \cellcolor{cmpbetter!19}$7.0$ & \cellcolor{cmpbetter!7}$3.3$ & $2.1$ & $2.5$ \\
\bottomrule

\end{tabular}
\caption{Variance-reduction ratio over Monte Carlo at equal cost,
$R = 200$ randomizations, at $N = 2^{13}$ and at $N = 1000$. The last
column gives the same sequence at its own size $N = 1024$. A ratio of ten
means one tenth of the Monte Carlo variance. A dagger marks a cell where
the $R$ estimates are dominated by a few of them and is known to an order
of magnitude. The others are known to between $8$ and $20\%$.}
\label{tab:emulator}
\end{table}

\subsection{Analysis of the convergence rate}
\label{sec:rate}

Before Latinization, LAT is a stratified design and converges at the
corresponding rate. The argument is a classical count of the cells on the
boundary \citep{pausinger2016,he2016}. The hypothesis on the diameter is
provided by the bound on the aspect ratio of Sect.~\ref{sec:golden}.

\begin{proposition}
\label{prop:rate}
Partition $[0,1]^d$ into $N$ cells of volume $1/N$ and diameter at most
$c\,N^{-1/d}$, and place one point uniformly at random in each. For an
axis-aligned box $B$, the error between the fraction of points in $B$ and
its volume has mean zero and standard deviation at most of order
$N^{-(d+1)/(2d)}$.
\end{proposition}

\begin{proof}
A cell that meets the boundary of $B$ lies within $c\,N^{-1/d}$ of it, so
the cells concerned are contained in a slab of volume at most
$4dc\,N^{-1/d}$ around $\partial B$ and there are at most
$4dc\,N^{(d-1)/d}$ of them. A cell inside $B$ contributes zero and a cell
outside as well, and each boundary cell contributes an independent term
of size $1/N$, whence a variance $O(N^{(d-1)/d} N^{-2}) =
O(N^{-(d+1)/d})$.
\end{proof}

The ASD is the squared error of an anchored box, averaged over the boxes
and over the $2^d$ corners, so Proposition~\ref{prop:rate} bounds its
expectation directly. The exponent is $-0.75$ at $d = 2$, $-0.60$ at
$d = 5$, $-0.55$ at $d = 10$ and $-0.52$ at $d = 30$, and the measured
slopes of LAT on the ASD are $-0.78$, $-0.58$, $-0.52$ and $-0.50$. The
higher dimensions match the prediction closely. In $d = 2$ the measured
slope is steeper than the bound, since the Latin margins remove a part of
the error that the boundary count does not model. This rate lies between
the $-0.5$ of Monte Carlo and the almost $-1$ of a digital net. The nets
reach the faster rate because they have one point per cell at each dyadic
scale simultaneously, so the boundary errors cancel instead of averaging
out \citep{niederreiter1992}. A partition that is not a grid has no such
global alignment.

\section{Refinement, optimization and constrained regions}
\label{sec:sequential}

A design is rarely used only once for an integral. The analyst adds runs
where the response is of interest, searches for an optimum, or has to
respect a design region that is not a box. The cells of LAT support these
three uses directly. They form a partition of the cube with exact
volumes, every point lies inside its cell, and each cell can be generated
alone from its index as shown in Sect.~\ref{sec:golden}.

\subsection{Local refinement}
\label{sec:refine}

Adaptive designs place new runs where a surrogate model is uncertain
\citep{garud2017,liu2018,fuhg2021}, and refined stratified sampling
subdivides existing strata for this purpose \citep{shields2015}. With LAT
the refinement is done on the tiling itself. Each cell that meets a
region of interest is subdivided by the golden rule into $r$ sub-cells,
with one point per sub-cell, see Fig.~\ref{fig:refinement}. The
realization of the tiling is fixed by its seed and each cell has its own
random stream. Enlarging the region or the subdivision factor thus never
displaces a point placed elsewhere. The exact volumes of the cells are
quadrature weights, and the refined design still integrates without bias.
Table~\ref{tab:refine} checks both properties. The density inside the
region grows by more than an order of magnitude while the rest of the
design is preserved. Contrary to a dyadic refinement that multiplies the
sample by $2^d$, the tiling adds points a few at a time. Applied to the
whole hypercube, this subdivision is also the only sequential mechanism
LAT offers. The tilings for $N$ and $N+1$ differ, and the matching to
$N+1$ bins moves all the points when one is added.

\begin{figure}[tbp]
\centering
\includegraphics[width=0.82\textwidth]{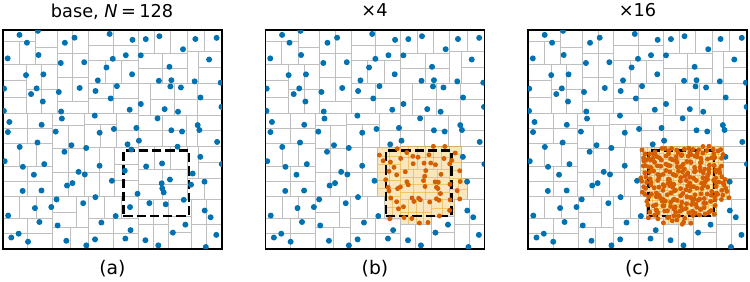}
\caption{Local refinement in $d = 2$. Cells meeting the target region
(dashed) are subdivided by the golden rule, by a factor four in (b) and
sixteen in (c). Points in cells that the region does not meet are
unchanged, bit for bit, across the three panels.}
\label{fig:refinement}
\end{figure}

\begin{table}[tbp]
\centering
\small
\begin{tabular}{lcccc}
\toprule
factor & points in region & points total & displaced outside & weighted bias \\
\midrule
$\times1$ (base) & 1.3 & 512 & 0 & $4.06\!\times\!10^{-6}$ $\pm$ $4.08\!\times\!10^{-5}$ \\
$\times4$ & 5.2 & 819 & 0 & $7.20\!\times\!10^{-6}$ $\pm$ $4.03\!\times\!10^{-5}$ \\
$\times16$ & 20.5 & 2046 & 0 & $7.97\!\times\!10^{-6}$ $\pm$ $3.99\!\times\!10^{-5}$ \\
\bottomrule

\end{tabular}
\caption{Refinement at $d = 5$ from a base design of $512$ points, region of
side $0.3$ at the centre, 99 randomizations. The integrand is
$\prod_j x_j^2$, of known integral $3^{-d}$. The $\pm$ term is the standard
error of the bias estimate.}
\label{tab:refine}
\end{table}

\subsection{Global optimization}
\label{sec:optim}

The same subdivision turns the tiling into a global optimizer of the
DIRECT family (\emph{dividing rectangles}, \citealp{jones1993}). DIRECT
also searches by splitting boxes. The search starts with one evaluation
per cell. Then, the cells holding the lowest observed values are split
with the golden rule and one jittered point is evaluated in each child.
This is repeated until the budget is spent. The split of a cell costs
$O(d)$ and does not touch the other cells, so the search has no overhead
compared to the evaluations.

Table~\ref{tab:optim} reports this best-first search with a budget of
$1024$ evaluations on the Branin and Hartmann-6 test problems
\citep{surjanovic2013}. It is compared to uniform random search, to
one-shot Sobol', LHS and LAT designs using the whole budget, and to
DIRECT itself as implemented in \texttt{scipy.optimize}. The score is the
regret and its distribution is strongly right-skewed. Two observations
can be made. First, the gain comes from the adaptivity and not from the
starting design. The refined search solves Branin in every repetition,
whereas the one-shot designs, LAT included, fail in $96\%$ of them or
more. On Hartmann-6 it reduces the median regret by a factor eight and
the failure rate from $100\%$ to $93\%$. Second, the comparison with
DIRECT depends on the dimension. The tiling search is better on Branin.
DIRECT is far better on Hartmann-6, where a random split wastes more
evaluations as the dimension grows. The tiling search is thus not meant
to replace DIRECT, nor the model-based methods that make a better use of
each run when the evaluations dominate the cost \citep{jones1998}. Its
interest is that a single partition serves as a design, as a quadrature
rule and as a search.

\begin{table}[tbp]
\centering
\small
\begin{tabular}{lcccccc}
\toprule
 & \multicolumn{3}{c}{Branin ($d = 2$)}
 & \multicolumn{3}{c}{Hartmann-6 ($d = 6$)} \\
\cmidrule(lr){2-4}\cmidrule(lr){5-7}
method & median & $q_{95}$ & fail & median & $q_{95}$ & fail \\
\midrule
random search & $2.82\!\times\!10^{-2}$ & $1.25\!\times\!10^{-1}$ & $98\%$ & $6.27\!\times\!10^{-1}$ & $9.07\!\times\!10^{-1}$ & $100\%$ \\
LHS (one shot) & $2.63\!\times\!10^{-2}$ & $1.15\!\times\!10^{-1}$ & $98\%$ & $5.63\!\times\!10^{-1}$ & $8.67\!\times\!10^{-1}$ & $100\%$ \\
Sobol' (one shot) & $3.24\!\times\!10^{-2}$ & $1.12\!\times\!10^{-1}$ & $96\%$ & $5.51\!\times\!10^{-1}$ & $9.00\!\times\!10^{-1}$ & $100\%$ \\
LAT (one shot) & $3.51\!\times\!10^{-2}$ & $1.20\!\times\!10^{-1}$ & $99\%$ & $6.06\!\times\!10^{-1}$ & $9.45\!\times\!10^{-1}$ & $100\%$ \\
\addlinespace
LAT refine & $0$ & $2.31\!\times\!10^{-15}$ & $0\%$ & $7.94\!\times\!10^{-2}$ & $2.28\!\times\!10^{-1}$ & $93\%$ \\
DIRECT & $9.03\!\times\!10^{-8}$ & $9.03\!\times\!10^{-8}$ & $0\%$ & $3.95\!\times\!10^{-5}$ & $3.95\!\times\!10^{-5}$ & $0\%$ \\
\bottomrule

\end{tabular}
\caption{Minimization at a fixed budget, 99 repetitions. LAT refine
starts from $64$ cells and repeatedly splits the $8$ best cells into $4$
children each, one jittered point per child, while the one-shot designs
use the whole budget at once. DIRECT is deterministic, so its three
columns hold the same value, and \texttt{scipy} stops it at $1037$
evaluations on Branin and at $733$ on Hartmann-6, where its default
tolerance is met first. ``Fail'' is the fraction of
repetitions ending with a regret above $10^{-3}$.}
\label{tab:optim}
\end{table}

\subsection{Non-rectangular regions}
\label{sec:domains}

When the design region is not a box, the classical options are either to
reject the points falling outside of it, or to optimize a design directly
on the region as done by fast flexible filling \citep{lekivetz2015}. The
kernel-density design of \citet{roy2020} includes the constraint in the
sampling probability. The tiling provides a direct construction. The
hypercube is partitioned finely enough for about $N$ cell centres to lie
in the region $\Omega$, and these cells are kept with one point each, see
Fig.~\ref{fig:domains}. The point of a cell that straddles the boundary
is redrawn inside its cell until it falls in $\Omega$, up to eight times,
and it is set to the centre otherwise. The number of kept cells estimates
the volume of $\Omega$. The estimator of a mean is biased, since the kept
cells enter with equal weights although their intersections with $\Omega$
differ. At $d = 5$ on the ball, around one fifth of the kept points need
at least one redraw.

Table~\ref{tab:domains} compares the integration of the mean value over a
disk and a ball with the rejection of a scrambled Sobol' or Monte Carlo
stream down to the same number of accepted points. The result depends on
the dimension through the thickness of the boundary layer. In two
dimensions the filtering is at least as good as rejection. In five
dimensions it is several times worse. A large fraction of a
high-dimensional ball lies within one cell of the boundary, where the
equal-volume covering is coarse. As it stands, the construction is thus
an alternative to rejection in low dimension only. Weighting the boundary
cells by their intersection with $\Omega$ would remove the bias and most
of the difference at $d = 5$. Exact recursive splits of non-rectangular
sets exist for special sample sizes. The scrambled geometric nets of
\citet{basuowen2017} split triangles and their products with proven
rates. The filtering of the cells works for arbitrary domains and any
sample size.

\begin{figure}[tbp]
\centering
\includegraphics[width=0.82\textwidth]{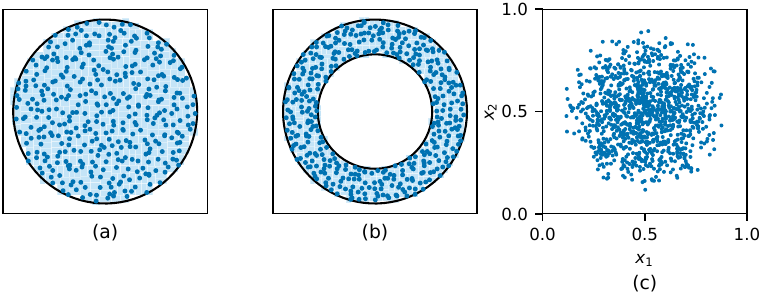}
\caption{Sampling non-rectangular regions by keeping the cells with
their centre inside. (a) Disk and (b) annulus in $d = 2$ with the kept
cells shown. (c) The $(x_1, x_2)$ projection of a design on a ball in $d = 5$. All the
points lie inside.}
\label{fig:domains}
\end{figure}

\begin{table}[tbp]
\centering
\small
\begin{tabular}{lcccc}
\toprule
 & \multicolumn{2}{c}{$d = 2$ (disk)} & \multicolumn{2}{c}{$d = 5$ (ball)} \\
\cmidrule(lr){2-3}\cmidrule(lr){4-5}
method & error & shell dev. & error & shell dev. \\
\midrule
LAT (tile filtering) & $3.45\!\times\!10^{-4}$ & 0.064 & $3.83\!\times\!10^{-3}$ & 0.298 \\
Sobol' (rejection) & \cellcolor{cmpworse!6}$4.26\!\times\!10^{-4}$ & \cellcolor{cmpworse!6}0.066 & \cellcolor{cmpbetter!40}$7.26\!\times\!10^{-4}$ & \cellcolor{cmpbetter!40}0.134 \\
MC (rejection) & \cellcolor{cmpworse!40}$1.85\!\times\!10^{-3}$ & \cellcolor{cmpworse!40}0.146 & \cellcolor{cmpbetter!30}$1.10\!\times\!10^{-3}$ & \cellcolor{cmpbetter!36}0.146 \\
\bottomrule

\end{tabular}
\caption{Mean-value integration of $\lVert x - c\rVert^2$ over the ball,
about $1024$ accepted points, 99 randomizations. Shell deviation is the
largest relative count deviation over eight equal-volume radial shells.}
\label{tab:domains}
\end{table}

\section{Conclusions and perspectives}
\label{sec:conclusion}

This work proposes a new method to build space-filling designs for an
arbitrary number of samples, referred to as LAT. This is a three-step
process: (i) the unit hypercube is partitioned into $N$ cells of equal
volume by a recursion following the golden section, (ii) one point is
placed in each cell, and (iii) the margins are Latinized by a matching of
the cells to the Latin bins, every point staying inside its cell. The
partition is defined and analysed for any $N$. Its recursion follows the
Fibonacci word, its infinite limit is aperiodic, and the sizes at which
the cells align on a grid are identified. The claims of the cells over the
Latin bins form a doubly stochastic matrix. The margins can thus be made
Latin with every point left uniform inside its cell and the estimator
unbiased.

Several findings stand out. The construction is simple. It only needs
integer splits, the cells of a level are cut in one vectorized operation,
and any cell can be generated alone from its index. The Latinization is of
first order and the cut rule of second order. Under a common matching the
golden cut is ahead of the balanced one, the more so in low dimension.
Aperiodicity makes the construction analysable and is not by itself a
source of accuracy. LAT is better than Monte Carlo in all dimensions
considered, and better than LHS as soon as the integrand has interactions.
The gain grows with the order of the interactions and fades in high
dimension, where all the methods converge at the Monte Carlo rate. On the
average squared discrepancy of \citet{clement2025}, LAT lies between the
scrambled sequences and LHS in moderate dimension, and all the designs
meet in high dimension. The accuracy of LAT does not depend on the sample size,
whereas a scrambled net truncated away from a power of two loses one to
two orders of magnitude on smooth functions. The nets remain better at the
powers of two, and a randomly shifted lattice rule on periodic integrands.

The cells are what makes LAT more than a design. A design can be refined
inside a region of interest without displacing a point placed elsewhere.
The exact volumes of the cells are quadrature weights, so the refined
design still integrates without bias. Splitting the best cells gives a
global search at no cost beyond the evaluations. Keeping the cells with
their centre inside a region samples a non-rectangular domain, and this
is competitive with rejection in low dimension. The quality of the design
is of prime importance as it determines the quality of the analysis of
the experiments. The proposed method provides an alternative to truncated
sequences, and to optimized designs when a design has to be rebuilt many
times or when the sample size is too large for a search.

Perspective for this work stands in the use that can be made of the cells
themselves. The integral over a sub-domain is the weighted sum of the
cells it covers, and the dispersion inside a cell is estimated from the
runs it holds. Further runs can then be allocated to the cells that
contribute the most to the variance, in the spirit of recursive stratified
integration \citep{press1990}. As the refinement leaves the points outside
the region untouched, this allocation can be decided once the campaign has
started. A second direction is the input distribution. Engineering inputs
are rarely uniform, and the construction only needs a rule for the number
of points a cell receives. Allocating the children by probability mass
instead of volume would stratify the input measure directly. The same rule
accommodates a design region defined by constraints, and weighting the
boundary cells would make the sampling of such regions unbiased. Finally,
the exact matching costs more than the greedy one, and an exact
Latinization at the cost of a sort would make the unbiased variant the
default. Estimating Sobol' indices, a cross-validation loop or a restarted
optimization all rebuild designs at whatever size the budget leaves
\citep{tissot2015}, and this is where a design defined at any size is most
useful.

\bibliographystyle{abbrvnat}
\bibliography{refs}

\end{document}